\documentclass[12pt,english]{article}
\usepackage[margin=1.5in]{geometry}

\usepackage{mdwlist}
\usepackage{enumerate}
\usepackage{amssymb,amsbsy,latexsym}
\usepackage{amsmath}
\usepackage{graphics, subfigure, float}
\usepackage{fp, calc}
\usepackage{hyperref}
\usepackage{url}

\usepackage[T1]{fontenc} 
\usepackage{fourier}
\usepackage{bm}

\usepackage{amscd,amsthm}

\usepackage{pst-all}
\usepackage{pstricks-add}
\usepackage{pst-func}
\newpsobject{showgrid}{psgrid}{subgriddiv=1,griddots=10,gridlabels=6pt}
\usepackage{verbatim, comment}
\usepackage{datetime}

\newtheoremstyle{theorem}{1em}{1em}{\slshape}{0pt}{\bfseries}{.}{ }{}
\theoremstyle{theorem}
\newtheorem{theorem}{Theorem}

\newtheorem{conjecture}{Conjecture}
\newtheorem*{theorem*}{Theorem}
\newtheorem{corollary}[theorem]{Corollary}
\newtheorem*{corollary*}{Corollary}
\newtheorem{proposition}[theorem]{Proposition}
\newtheorem{lemma}[theorem]{Lemma}

\newtheorem*{claim*}{Claim}
\theoremstyle{remark}
\newtheorem{remark}{Remark}

\newtheorem*{remark*}{Remark}

\providecommand{\setN}{\mathbb{N}}

\providecommand{\setR}{\mathbb{R}}
\newcommand{\E}{\mathop{\mathbb{E}}}

\newcommand{\eps}{\varepsilon}

\usepackage{calrsfs} 
\DeclareMathAlphabet{\pazocal}{OMS}{zplm}{m}{n}

\usepackage[displaymath,textmath,graphics, subfigure, floats]{preview} 
\PreviewEnvironment{center} 
\PreviewEnvironment{pspicture} 

\makeatother

\title{Optimal Online Discrepancy Minimization}
\author{Janardhan Kulkarni\thanks{Microsoft Research, Redmond. Email: {\tt jakul@microsoft.com}.} \; \and \; Victor Reis\thanks{University of Washington, Seattle. Email: {\tt voreis@uw.edu}.} \;\; and \; Thomas Rothvoss\thanks{University of Washington, Seattle. Email: {\tt rothvoss@uw.edu}. Supported by NSF CAREER grant 1651861, NSF grant 2318620 and a David \& Lucile Packard Foundation Fellowship.}}
\date{}

\begin{document}

\maketitle

\begin{abstract}
  We prove that there exists an online algorithm that for any sequence of vectors $v_1,\ldots,v_T \in \setR^n$
  with $\|v_i\|_2 \leq 1$, arriving one at a time, decides random signs $x_1,\ldots,x_T \in \{ -1,1\}$ so that
  for every $t \le T$, the prefix sum $\sum_{i=1}^t x_iv_i$ is $10$-subgaussian. This improves over the work of
  Alweiss, Liu and Sawhney who kept prefix sums $O(\sqrt{\log (nT)})$-subgaussian, and gives a $O(\sqrt{\log T})$ bound on the discrepancy $\max_{t \in T} \|\sum_{i=1}^t x_i v_i\|_\infty$. 
  Our proof combines a generalization of Banaszczyk's prefix balancing result to trees with a cloning argument to find distributions rather than single colorings. We also show a matching $\Omega(\sqrt{\log T})$ strategy for an oblivious adversary.
\end{abstract}

\section{Introduction}

We study \emph{online} vector balancing problems, first considered  by Spencer in late 70's \cite{DBLP:journals/jct/Spencer77}.
We receive vectors $v_1,\ldots,v_T \in \setR^n$, which are bounded in some norm, one at a time, and we have to decide the sign $x_i \in \{ -1,1\}$
{\em irrevocably} after learning the vector $v_i$. 
The goal is to keep the signed sum $\sum_{i=1}^T x_iv_i$ small in some norm; a natural variant asks that all prefixes $\sum_{i=1}^t x_i v_i$ are small for all $t \in [T]$.
This vector balancing formulation captures several classic problems in discrepancy theory, where the norm to be balanced is the maximum absolute value of any coordinate, also known as the $\ell_\infty$ norm.

If $\|v_i\|_{\infty} \leq 1$, uniformly random signs achieve a $O(\sqrt{T \log n})$ $\ell_\infty$ bound, and there are several methods
to make this a deterministic online algorithm, see for example the excellent book
of Chazelle~\cite{DiscrepancyMethod-Chazelle-2000}.
Unfortunately, the random coloring is tight in its dependency  on $\sqrt{T}$ for $n \ge 2$ \cite{DBLP:journals/jct/Spencer77, TenLectures}.
Indeed, an {\em adaptive adversary} can simply choose a vector $v_t \in [-1,1]^n$ that is orthogonal to the current position $\sum_{i=1}^{t-1} x_iv_i$ and satisfies $\|v_t\|_2^2 \ge n-1$, then $\|\sum_{i=1}^t x_i v_i\|_2 \ge \sqrt{(n-1)T}$ and, in particular, the discrepancy is $\Omega(\sqrt{T})$.

Much of the focus of subsequent efforts has been to improve the dependence on the $\sqrt{T}$ term by restricting the power of the adversary.
One natural choice is to consider the stochastic setting. 
Here vectors $v_i$ are sampled independently from a distribution $\bold{p}$ that is known to the online algorithm.
When $\bold{p}$ is a uniform distribution on all $\{-1, 1 \}^n$ vectors, Bansal and Spencer \cite{BansalSpencer2020} showed that one can get $O(\sqrt n)$ discrepancy for the $\ell_{\infty}$-norm, or $O(\sqrt{n} \log T)$ for all prefixes up to time $T$.
Motivated by the applications of online discrepancy minimization techniques to online envy minimization problems \cite{jiang2019online}, 
Bansal, Jiang, Singla and Sinha~\cite{DBLP:journals/corr/abs-1912-03350} considered general distributions $\bold{p}$ supported on $[-1, 1]^n$.
For this problem, they achieved an $\ell_{\infty}$-discrepancy of $O(n^2 \log (nT))$, and this was improved by Bansal, Jiang, Meka, Singla, and Sinha \cite{bansal2021online} who showed a bound of $O(\sqrt{n} \log^4(nT))$. The work of Aru, Narayanan, Scott and Venkatesan \cite{aru2018balancing} achieved a bound of $O_n(\sqrt{\log T})$, where the dependence on $n$ is super exponential. One important point to note in all of these results is that they substantially improve the dependence on $\Omega(\sqrt{T})$ to logarithmic factors.

Despite these impressive results, until very recently, little progress was made on the online vector balancing problem against {\em oblivious} adversaries  -- the most common setting considered in the online algorithms literature.  
Here the adversary fixes an arbitrary set of vectors  $v_1,\ldots,v_T \in \setR^n$ in advance, and the online algorithm can use randomized strategies. For the special case of \emph{edge orientation}, where the vectors correspond to columns of the incidence matrix of a graph, a simple random labeling argument due to Kalai~\cite{NAOR200593} achieves a discrepancy bound of $O(\log T)$ with high probability. 

In an elegant result, Alweiss, Liu and Sawhney~\cite{SelfBalancingRandomWalkSTOC2021} showed that a very simple
\emph{self-balancing random walk} can find random signs so that with high probability all prefixes $\sum_{i=1}^t x_iv_i$ are $O(\sqrt{\log(nT)})$-\emph{subgaussian}. Here a random vector $X \in \setR^n$ is called \emph{$c$-subgaussian} if
for any unit direction $w \in S^{n-1}$ one has $\E[\exp(\left<X,w\right>^2/c^2)] \le 2$.
In particular, $\|\sum_{i=1}^t x_iv_i\|_{\infty} \leq O(\log(nT))$ for all $t \in [T]$ with high probability against any oblivious adversary.

\subsection{Our contributions}

Our main contribution is:
\begin{theorem} \label{thm:MainResultI}
There is an online algorithm that against any oblivious adversary and for any sequence of vectors $v_1,\ldots,v_T \in \setR^n$ with $\|v_i\|_2 \leq 1$, arriving one at a time, decides random signs $x_1,\ldots,x_T \in \{ -1,1\}$ so that for every $t \in [T]$, the prefix sum $\sum_{i=1}^t x_iv_i$ is $10$-subgaussian.
\end{theorem}
The algorithm does not depend on $T$ so one may also take an infinite sequence of
vectors.  

By using the machinery of Talagrand's majorizing measures theorem, we may recover Banaszczyk's theorem in the online setting.

\begin{theorem} \label{thm:MainResultBanaszczyk}
 Given a symmetric convex body $K \subseteq \setR^n$, there is an online algorithm that against any oblivious adversary and for any sequence of vectors $v_1,\ldots,v_T \in \setR^n$ with $\|v_i\|_{2} \leq 1$, arriving one at a time, decides random signs $x_1,\ldots,x_T \in \{ -1,1\}$ so that each of the following hold with probability at least $1/2$:
\begin{enumerate*}
\item[(a)] $\sum_{i=1}^T x_iv_i \in O(1) \cdot K$ under the assumption $\gamma_n(K) \ge \frac{1}{2}$.
\item[(b)] $\sum_{i=1}^t x_iv_i \in O(1) \cdot K$ for all $t \in [T]$ under the assumption $\gamma_n(K) \ge 1 - \frac{1}{2T}$.
\end{enumerate*}
\end{theorem}

For $\ell_p$ discrepancy minimization, we obtain the following corollary:

\begin{corollary}\label{thm:MainResultLp}
There is an online algorithm that against any oblivious adversary and for any sequence of vectors $v_1,\ldots,v_T \in \setR^n$ with $\|v_i\|_{2} \leq 1$, arriving one at a time, decides random signs $x_1,\ldots,x_T \in \{ -1,1\}$ so that each of the following hold with probability at least $1- \delta$ for any $\delta \in (0,\frac{1}{2}]$ and any $p \ge 2$:
\begin{enumerate*}
\item[(a)] $\|\sum_{i=1}^T x_iv_i\|_p \lesssim \sqrt{p} \min(n,T)^{1/p} + \sqrt{\log(1/\delta)}$;
\item[(b)] $\max_{t \in [T]} \|\sum_{i=1}^t x_iv_i\|_p \lesssim \sqrt{p} \min(n,T)^{1/p} + \sqrt{\log T} + \sqrt{\log(1/\delta)}$.
\end{enumerate*}
Furthermore,
\begin{enumerate*}
\item[(c)] $\|\sum_{i=1}^T x_iv_i\|_\infty \lesssim \sqrt{\log \min(n,T)} + \sqrt{\log(1/\delta)}$;
\item[(d)] $\max_{t \in [T]} \|\sum_{i=1}^t x_iv_i\|_\infty \lesssim \sqrt{\log T} + \sqrt{\log(1/\delta)}$.
\end{enumerate*}
\end{corollary}

Our bounds match the best known upper bounds in the offline setting where all the vectors are known in advance. We show Corollary~\ref{thm:MainResultLp}(d) is tight in the oblivious setting even when $n = 2$:

\begin{theorem} \label{thm:OnlineDiscLB}
For any $n \ge 2$, there is a strategy for an oblivious adversary that yields a sequence of unit vectors $v_1, \dots, v_T \in \setR^n$ so that for any online algorithm, with probability at least $1 - 2^{-T^{\Omega(1)}}$, one has $\max_{t \in [T]} \|\sum_{i=1}^t x_iv_i\|_\infty \gtrsim \sqrt{\log T}$.
\end{theorem}

This improves upon the $\Omega\Big(\sqrt{\tfrac{\log T}{\log \log T}}\Big)$ lower bound of~\cite{DBLP:journals/corr/abs-1912-03350}. For the $\ell_2$ norm, Corollary~\ref{thm:MainResultLp}(a) is tight for any $n, T$, since an orthonormal basis followed by zeros does achieve a lower bound of $\sqrt{\min(n,T)}$. If the orthonormal basis is instead followed by the construction in~\ref{thm:OnlineDiscLB}, we get a matching lower bound construction for Corollary~\ref{thm:MainResultLp}(b). On the other hand, it is not known whether there exists a lower bound better than constant for Corollary~\ref{thm:MainResultLp}(c).

Finally, our framework also provides improved bounds for the edge orientation problem.

\begin{corollary} [Online edge orientation] \label{cor:edgeOrientation}
There exists an online algorithm that for any set of $n$ vertices and any sequence of edges, arriving one at a time, decides orientations so that at every vertex, the absolute difference between indegree and outdegree always remains bounded by $O(\sqrt{\log T})$ with high probability.
\end{corollary}

\subsection{An Overview of our Techniques}

Suppose we play against an oblivious adversary that has a predetermined sequence of vectors $v_1,\ldots,v_T \in \setR^n$
with $\|v_i\|_2 \leq 1$ for all $i \in [T]$ that are revealed to us one vector at a time and we need to determine
random signs $x_i \in \{ -1,1\}$. After some discretization one may
think of this game as a balancing problem on a rooted tree $\pazocal{T} = (V,E)$ where edges $e \in E$ are labelled with vectors $v_e$
of length at most one.
The adversary chooses a path from the root to a leaf which is revealed one edge at a time.
After learning the next edge on the path we must give it a random sign $x_e \in \{ -1,1\}$ to keep the sum $\sum_{e \in P} x_ev_e$
subgaussian where $P$ is the path chosen so far by the adversary. Here the tree $\pazocal{T}$ has depth $T$ but each
interior node will have an outgoing edge for each vector in a fine enough $\varepsilon$-net and so degrees in $\pazocal{T}$ are
exponential in $n$ and $T$.

The next observation is that rather than only drawing signs for the selected path we can
draw all signs $x \in \{ -1,1\}^E$ with the goal of keeping  $\sum_{e \in P} x_ev_e$  $O(1)$-subgaussian for \emph{all} root-node paths $P$. This is not actually a restriction since the adversary could pick any such path anyway.
Note that indeed this argument assumes an oblivious adversary.
In order to find the distribution we want to make use of the following powerful result:
\begin{theorem}[Banaszczyk~\cite{BalancingVectorsBanaszczyk1998}] \label{thm:BanaszczykStretchedBodyLB} 
There exists a constant\footnote{Banaszczyk's proof works as long as $\int_{-\beta}^\beta e^{-t^2/2} dt < \int_1^\infty e^{-t^2/2} dt$; for example take $\beta := 0.2001$.} $\beta > \frac{1}{5}$ so that for any convex body $K \subseteq \setR^n$ with $\gamma_n(K) \geq \frac{1}{2}$ and vector $u \in \setR^n$ with $\|u\|_2 \leq \beta$, there is a convex body $(K * u) \subseteq (K + u) \cup (K-u)$ with $\gamma_n(K * u) \geq \gamma_n(K)$.
\end{theorem}
Banaszczyk's construction for the body $K * u$ is quite intuitive: call a line $x + \setR u$ \emph{long} if its intersection with $K$ has length at least $2\|u\|_2$. Then intersect $(K + u) \cup (K-u)$ with all long lines. However, the proof of the inequality $\gamma_n(K * u) \geq \gamma_n(K)$ is quite skillful. 
The main application in Banaszczyk~\cite{BalancingVectorsBanaszczyk1998} was to prove that
for any vectors $v_1,\ldots,v_T \in \setR^n$ with $\|v_i\|_2 \leq 1$ and any convex body
$K \subseteq \setR^n$ with $\gamma_n(K) \geq 1/2$, there are signs $x \in \{ -1,1\}^T$ with $\sum_{i=1}^T x_iv_i \in 5K$.
In a later work, Banaszczyk applied his Theorem~\ref{thm:BanaszczykStretchedBodyLB} also to the prefix setting:

\begin{theorem}[Banaszczyk~\cite{SeriesOfSignedVectorsBanaszczyk2012}] \label{thm:BanaszczykPrefixBalancing}
  There is a constant $\alpha < 5$, 
  so that for any $v_1,\ldots,v_T \in \setR^n$ with $\|v_i\|_2 \leq 1$ for $i=1,\ldots,T$ and any convex body $K \subseteq \setR^n$ with $\gamma_n(K) \geq 1-\frac{1}{2T}$, there are signs $x_1,\ldots,x_T \in \{ -1,1\}$ so that
  \[
   \sum_{i=1}^t x_i v_i \in \alpha K \quad \forall t=1,\ldots,T.
  \]
\end{theorem}
One can think of this result as balancing all root-node paths on a path graph into a body $K$.
It is also interesting to note that the construction behind Theorem~\ref{thm:BanaszczykPrefixBalancing} is inherently non-online --- the sign for $v_i$ may depend on all other vectors $v_1,\ldots,v_{i-1}$ and $v_{i+1},\ldots,v_{T}$.  
Our first step is to generalize Theorem~\ref{thm:BanaszczykPrefixBalancing} to trees $\pazocal{T} = (V,E)$ and prove that for any convex body $K$ with   $\gamma_n(K) \geq 1-\frac{1}{2|E|}$ we can find
signs $x \in \{ -1,1\}^E$ so that $\sum_{e \in P} x_ev_e \in \alpha K$ for any root-node path $P$.
However, there are two problems with that approach:
\begin{enumerate}
\item[(i)] the tree $\pazocal{T}$ obtained in the reduction
  has exponential size and $K$ would need to be huge to satisfy   $\gamma_n(K) \geq 1-\frac{1}{2|E|}$;
\item[(ii)] a straightforward application of
  Banaszczyk's framework will only produce a single vector of signs $x \in \{ -1,1\}^E$ rather than a distribution.
\end{enumerate}
Interestingly, both problems can be solved with the same \emph{cloning trick}: we replace each edge $e$ with a sequence of $N$ edges labelled with orthogonal copies of $v_e$, say $v_e^{(1)},\ldots,v_e^{(N)}$. 
Then the signs $x_e^{(\ell)} \in \{ -1,1\}$ for the blown-up tree can be turned into a distribution over signs $(x_e^{(\ell)})_{e \in E}$ for the original tree by drawing $\ell \sim \{ 1,\ldots,N\}$ uniformly. This solves issue (ii). Considering (i), our trick makes the tree even larger by a factor $N$, worsening the issue! However, choosing
\[
 K = \big\{ (y^{(\ell)}_e)_{e \in E, \ell \in [N]} \mid (y^{\ell}_e)_{e \in E}\textrm{ is }O(1)\textrm{-subgaussian when }\ell \sim [N] \big\}
\]
we can prove that the Gaussian measure of $K$ increases fast enough as $N$ grows to compensate
for the blowup of the tree.

\subsection{A Brief History of Algorithmic Discrepancy}
\label{sec:briefhistory}
\emph{Discrepancy theory} is a classical area of combinatorics where one is given vectors $v_1,\ldots,v_T \in \setR^n$ that are bounded
in some norm and the goal is to determine signs $x_1,\ldots,x_T \in \{ -1,1\}$ so that
the signed sum $\sum_{i=1}^T x_iv_i$ is again bounded. For example, the celebrated
\emph{Spencer's Theorem}~\cite{SixStandardDeviationsSuffice-Spencer1985} says that for $\|v_i\|_{\infty} \leq 1$, there are signs $x \in \{ -1,1\}^T$ so that
$\|\sum_{i=1}^T x_iv_i\|_{\infty} \leq O(\sqrt{T \log(2n/T)})$ when $n \geq T$ and 
$\|\sum_{i=1}^T x_iv_i\|_{\infty} \leq O(\sqrt{n})$ for $T \geq n$. 
Often this result is phrased as bi-coloring elements in a set system in which case
$(v_1,\ldots,v_T)$ is the incidence matrix of the set system and $v_i$ is the incidence
vector of element $i \in [T]$. Beck and Fiala~\cite{BECK19811}
proved that for any vectors $v_1,\ldots,v_T \in \setR^n$ with $\|v_i\|_1 \leq 1$ one can find signs $x \in \{ -1,1\}^T$ so that $\|\sum_{i=1}^T x_iv_i\|_{\infty} \leq 2$. Another seminal result that
will be important for our consideration is the one by Banaszczyk~\cite{BalancingVectorsBanaszczyk1998} who proved that for any
$v_1,\ldots,v_T \in \setR^n$ with $\|v_i\|_{2} \leq 1$ and any convex body $K \subseteq \setR^n$ with $\gamma_n(K) \geq 1/2$ there are signs with $\sum_{i=1}^T x_iv_i \in 5K$.

Several of the original arguments in discrepancy theory were nonconstructive in nature, in particular the pigeonhole principle used
by Spencer~\cite{SixStandardDeviationsSuffice-Spencer1985} and the convex geometric argument
of Banaszczyk~\cite{BalancingVectorsBanaszczyk1998}. Starting with the breakthrough of Bansal~\cite{DiscrepancyMinimization-Bansal-FOCS2010}, there has been a long sequence of work on algorithmic discrepancy~\cite{DiscrepancyMinimization-LovettMekaFOCS12,DBLP:journals/rsa/EldanS18,Levy2016DeterministicDM,DBLP:journals/rsa/ReisR23}. Bansal, Dadush, Garg and Lovett \cite{GramSchmidtWalk-BansalDGL-STOC18} finally found a polynomial time algorithm that makes the result of Banaszczyk \cite{BalancingVectorsBanaszczyk1998} constructive. To be more precise, \cite{GramSchmidtWalk-BansalDGL-STOC18} show that for any vectors
$v_1,\ldots,v_T \in \setR^n$ with $\|v_i\|_2 \leq 1$,  
there is an efficiently computable distribution over signs $x \in \{ -1,1\}^T$ so that
the random vector $\sum_{i=1}^T x_iv_i$ is $O(1)$-subgaussian, which by \cite{DBLP:conf/approx/DadushGLN16} is equivalent to the result of Banaszczyk.

Much of the relevance of discrepancy theory to theoretical computer science comes from
the close connection to rounding solutions to linear programs which is a key technique
in approximation algorithms. For example, Rothvoss and Hoberg~\cite{DBLP:conf/focs/Rothvoss13,DBLP:conf/soda/HobergR17} used discrepancy theory to approximate bin packing within an additive $O(\log OPT)$ and Bansal, Charikar, Krishnaswamy and Li \cite{DBLP:conf/soda/BansalCKL14} made use of it for an improved LP rounding for broadcast scheduling. More recently Bansal, Rohwedder and Svensson~\cite{FlowTimeAndPrefixBeckFialaSTOC2022} used the result of Banaszczyk~\cite{BalancingVectorsBanaszczyk1998,SeriesOfSignedVectorsBanaszczyk2012} to obtain better
bounds on the integrality gap of LPs for scheduling with flow-time objectives. Another application
of discrepancy lies in differential privacy, see for example the seminal paper by Nikolov, Talwar and Zhang~\cite{DBLP:conf/stoc/NikolovTZ13}.

A natural variant in discrepancy theory is the one of \emph{prefix discrepancy}, where
given vectors $v_1,\ldots,v_T$, the goal is to keep the prefix $\sum_{i=1}^t x_iv_i$ bounded for all $t \in [T]$. This is crucial, for example, in the scheduling application of Bansal et al~\cite{FlowTimeAndPrefixBeckFialaSTOC2022}. Some arguments such as the one of \cite{SixStandardDeviationsSuffice-Spencer1985} and the one of Banaszczyk~\cite{BalancingVectorsBanaszczyk1998,SeriesOfSignedVectorsBanaszczyk2012}
do extend to this setting, whereas it is open whether the result of Beck and Fiala~\cite{BECK19811} has a prefix analogue. 
Broadly speaking, online discrepancy minimization problems, in some sense, include the prefix discrepancy question because the adversary could stop at any moment.

\section{Preliminaries}

First, we review a few facts for later.

\subsection{Geometry}

A \emph{convex body} is a set $K \subseteq \setR^n$ that is convex, compact (closed and bounded) and full-dimensional.
Let $B_2^n := \{ x \in \setR^{n} \mid \|x\|_2 \leq 1\}$ be the \emph{Euclidean ball}
and let $S^{n-1} := \{ x \in \setR^n  \mid \|x\|_2 = 1\}$ be the \emph{sphere}.
A set $W \subseteq S^{n-1}$ is called an \emph{$\varepsilon$-net} if for all $x \in S^{n-1}$, there is a $y \in W$
with ${\|x-y\|_2 \leq \varepsilon}$. 
\begin{lemma} \label{lem:SizeOfEpsilonNet}
For any $0<\varepsilon \leq 1$, there is an $\varepsilon$-net $W \subseteq S^{n-1}$ of size $|W| \leq (\frac{3}{\varepsilon})^n$.
\end{lemma}
\begin{proof}
Pick any maximal set of points $W \subseteq S^{n-1}$ that have $\|\cdot\|_2$-distance at least $\varepsilon$ to each other. Then $W$ is an $\varepsilon$-net. Moreover the balls $x + \frac{\varepsilon}{2} B^n_2$ are disjoint for $x \in W$ and contained in $(1 + \frac{\varepsilon}{2})B^n_2$. Hence 
\[|W| \le \frac{\mathrm{Vol}_n ((1 + \frac{\varepsilon}{2}) \cdot B^n_2)}{\mathrm{Vol}_n (\frac{\varepsilon}{2} \cdot B^n_2)} = \left(\frac{1 + \frac{\varepsilon}{2}}{\frac{\varepsilon}{2}}\right)^n \le \left(\frac{3}{\varepsilon}\right)^n. \qedhere\]
\end{proof}
\begin{lemma} \label{lem:NetImpliesSmallConicComb}
Let $W \subseteq S^{n-1}$ be an $\varepsilon$-net for $0<\varepsilon<1$. Then for any $w_0 \in S^{n-1}$ there is
a coefficient vector $\lambda \in \setR_{\geq 0}^W$ with $w_0 = \sum_{w \in W} \lambda_w w$ and $\sum_{w \in W} \lambda_w \leq \frac{1}{1-\varepsilon}$.
\end{lemma}
\begin{proof}
  Consider $L(u) := \min\{ \|\lambda\|_1 : u = \sum_{w \in W} \lambda_w w\textrm{ and }\lambda_w \geq 0 \; \forall w \in W\}$.
  It is not entirely obvious, but $L(u)$ is well-defined for each $u \in S^{n-1}$. To see this, note that $L(u)$
  is the value of a linear program and if that program was infeasible then by Farkas Lemma (see e.g. Schrijver~\cite{TheoryOfLPandIPSchrijver1999})
  there is a $c \in S^{n-1}$ with $c^Tu > 0$ and $c^Tw < 0$ for all $w \in W$. Then $\|c-w\|_2 \geq 1$
  for all $w \in W$ which is a contradiction. 
  Now, fix a $w^* \in S^{n-1}$ maximizing the value $L(w^*)$; such a point must exist by compactness of $S^{n-1}$.
  Let $w \in W$ with $\|w^*-w\|_2 \leq \varepsilon$. Then writing $w^* = w + (w^*-w)$ we have
  \[
   L(w^*) \leq \underbrace{L(w)}_{\leq 1} + \underbrace{\|w^*-w\|_2}_{\leq \varepsilon} \cdot L\Big(\frac{w^*-w}{\|w^*-w\|_2}\Big) \leq 1 + \varepsilon L(w^*).
 \]
 Rearranging gives $L(w_0) \le L(w^*) \leq \frac{1}{1-\varepsilon}$.
\end{proof}

\subsection{Probability} 

We need to formalize what we mean that a random variable has Gaussian-type tails. There are several equivalent
ways to do so. We recommend the excellent textbook by Vershynin~\cite{HighDimProbabilityVershynin2018} for background.
\begin{proposition}[\cite{HighDimProbabilityVershynin2018}] \label{prop:EquivalentSubgaussianity}
  Let $X$ be a random variable. The following properties are equivalent; the parameters $C_i > 0$ appearing in these
  properties differ from each other by at most an absolute constant factor.
  \begin{enumerate}
  \item[(i)] One has $\Pr[|X| \geq t] \leq 2\exp(-t^2/C_1^2)$ for all $t \geq 0$.
  \item[(ii)] One has $\E[\exp(X^2 / C_2^2)] \leq 2$.
  \end{enumerate}
  If additionally $\E[X]=0$, then $(i)+(ii)$ are also equivalent to
  \begin{enumerate}
  \item[(iii)] One has $\exp(\lambda X) \leq \exp(C_3^2 \lambda^2)$ for all $\lambda \in \setR$.
  \end{enumerate}
\end{proposition}
While Dadush et al~\cite{DBLP:conf/approx/DadushGLN16} and Bansal et al~\cite{GramSchmidtWalk-BansalDGL-STOC18} made use of (iii), for
our purpose, it will be most useful to work with (ii). Hence for a real-valued random variable $X$, we define the \emph{subgaussian norm} as
\[
\| X\|_{\psi_2} := \inf \ \big\{t > 0 : \E[\exp(X^2/t^2)] \le 2\big\}.
\]
It can be proven that $\| \cdot \|_{\psi_2}$ is a norm on the space of jointly distributed random variables; indeed it is the \emph{Orlicz norm} associated with the (convex and nondecreasing) function $z \mapsto e^{z^2} - 1$. 
In particular $\|X + Y\|_{\psi_2} \leq \|X\|_{\psi_2} + \|Y\|_{\psi_2}$ for any random variables $X,Y$ --- even if they are not independent.

For the sake of completeness and in order to be able to give explicit constants, we show the directions
$(ii) \Rightarrow (i)+(iii)$ in Proposition~\ref{prop:EquivalentSubgaussianity}.  For example, the following lemma shows that any 10-subgaussian random variable $X$ satisfies $\E[\exp(\lambda X)] \le \exp(20\lambda^2)$. Therefore Theorem~\ref{thm:MainResultI} also yields the subgaussianity bound given in the main result of~\cite{GramSchmidtWalk-BansalDGL-STOC18}.
\begin{lemma} 
 If $\|X\|_{\psi_2} \le 1$ then $\Pr [|X| \ge t] \le 2e^{-t^2}$ for all $t \geq 0$. If moreover $\E[X] = 0$, then also $\E[\exp(\lambda X)] \le \exp(0.4 \lambda^2)$ for all $\lambda \in \setR$.
\end{lemma}

\begin{proof}
  The claim that $\Pr [|X| \ge t] \le 2e^{-t^2}$ follows directly from Markov's inequality. For the second claim, if $0.56\lambda^2 < 1$ we can use the inequality $e^{z} \le z + e^{0.56 z^2}$ (valid for all $z \in \mathbb{R}$) to obtain
  \[\E[\exp(\lambda X)] \le \E[\exp(X^2)^{0.56\lambda^2}] \stackrel{\textrm{concavity}}{\le} \underbrace{\E[\exp(X^2)]}_{\le 2 \ \textrm{as} \ \|X\|_{\psi_2} \le 1}{}^{0.56\lambda^2} \le \exp(0.4 \lambda^2).  \]
Else, $\lambda^2 \ge \tfrac{1}{0.56}$ and we will show that $e^{z} \le z + \tfrac{1}{2} \exp(0.4 \lambda^2 + \frac{z^2}{\lambda^2})$ for all $z \in \mathbb{R}$ to obtain
\[\E[\exp(\lambda X)] \le \tfrac{1}{2} \exp(0.4 \lambda^2) \E[\exp(X^2)] \le \exp(0.4 \lambda^2).\]
For $\lambda^2 = \tfrac{1}{0.56}$, the claim follows from $e^z \le z + e^{0.56z^2}$; the only other possible minimizer $\lambda_*$ of the right side occurs when $0.4 \lambda_*^2 = \frac{z^2}{\lambda_*^2}$, so $\lambda_*^4 = z^2/0.4 \ge \tfrac{1}{0.56^2}$. And indeed, we do have the inequality $e^z \le z + \tfrac{1}{2} \exp(2 \sqrt{0.4} |z|)$ for all $|z| \ge \tfrac{\sqrt{0.4}}{0.56}$.
\end{proof}

Next, we extend the concept of a subgaussian norm from random variables to random vectors.
For a random vector $X$ taking values in $\setR^n$, we denote 
\[
\|X\|_{\psi_2,\infty} := \sup_{w \in S^{n-1}} \|\langle X, w\rangle\|_{\psi_2}.
\] 
Again, $\| \cdot \|_{\psi_2,\infty}$ is a norm on the space of random vectors in $\setR^n$. For simplicity of notation, we sometimes write that $X$ is $t$-\emph{subgaussian} if $\|X\|_{\psi_2,\infty} \le t$.
Of course, a standard Gaussian $g \sim N(0,I_n)$ is also $O(1)$-subgaussian, but curiously not
with constant 1.
\begin{lemma} \label{lem:ExpOfGaussianSquare}
For $\lambda<\frac{1}{2}$ one has $\E_{g \sim N(0,1)}[\exp(\lambda \cdot g^2)] = \frac{1}{\sqrt{1-2\lambda}}$. In particular, the standard Gaussian $N(\bm{0},I_n)$ is $\sqrt{\tfrac{8}{3}}$-subgaussian.
\end{lemma}

We will also need the following upper bound on moments of sums of independent random variables:
\begin{lemma}[Rosenthal's inequality~\cite{RosenthalInequality}] \label{lem:RosenthalInequality}
Let $p \ge 2$ and let $X_1, \dots, X_N$ be mean zero independent random variables with finite $p$-th moment $\E [|X_i|^p]$. Then 
\[
\E[|X_1 + \dots + X_N|^p]^{1/p} \le 2^p \cdot \max\Big\{\Big(\sum_{i=1}^N \E[|X_i|^p]\Big)^{1/p},\Big(\sum_{i=1}^N \E[X_i^2]\Big)^{1/2} \Big\}.
\] 
\end{lemma}
We also mention that the constant $2^p$ can be improved to $\Theta(p/\log p)$~\cite{LatalaIneq}, though it will not be needed in our application.

Finally, we will use the tail bound form of Talagrand's comparison inequality. See once again the textbook of Vershynin~\cite{HighDimProbabilityVershynin2018} (Chapter 8.6).

\begin{lemma} \label{lem:TalagrandComparison}
Let $K \subseteq \setR^n$ be a symmetric convex body and $X \in \setR^n$ a random vector that is $O(1)$-subgaussian. Then with probability at least $1- \delta$ for any $\delta \in (0, \frac{1}{2}]$,
\[\|X\|_K \lesssim \E_{g \sim N(\bm{0}, I_n)}[\|g\|_K] + \sqrt{\log(1/\delta)} \cdot \frac{1}{\mathrm{inradius}(K)}, \]
where $\mathrm{inradius}(K)$ is the largest $r > 0$ so that $rB^n_2 \subseteq K$. 
\end{lemma}

\section{Generalizing Banaszczyk's prefix balancing argument to trees}

We begin by generalizing Banaszczyk's Theorem~\ref{thm:BanaszczykPrefixBalancing} to trees. 
A similar statement may also be found (with vectors assigned to vertices) as Theorem 4.1. in~\cite{Bansal2021PrefixDS}.
\begin{theorem} \label{thm:TreeVectorBalancing}
There exists a constant $\alpha < 5$ such that the following holds. Let $\pazocal{T} = (V,E)$ be a tree with a distinguished root $r \in V$ and $|E| \geq 1$, where each edge $e \in E$ is assigned a vector $v_e \in \setR^n$ with $\|v_e\|_2 \leq 1$.
Let $K \subseteq \setR^n$ be a convex body with $\gamma_{n}(K) \geq 1- \frac{1}{2|E|}$. Then there are signs
$x \in \{ -1,1\}^E$ so that 
\[
  \sum_{e \in P_i} x_e v_e \in \alpha K \quad \forall i \in V 
\]
where $P_i \subseteq E$ are the edges on the path from the root to $i$.
\end{theorem}

\begin{proof}
  Let $\beta > 1/5$ be the constant from Theorem~\ref{thm:BanaszczykStretchedBodyLB}.
  We write $C_i \subseteq V$ as all the \emph{children} of $i$ where $C_i = \emptyset$ for leaves
  and we write $D_i \subseteq V$ as the \emph{descendants} of $i$ (including $i$ itself).
  In particular $\{ i\} \cup C_i \subseteq D_i$. We construct a sequence of target bodies
  by setting
  \[
     K_i := \Big(\bigcap_{j \in C_i} (K_j * \beta v_{\{i,j\}})\Big) \cap K.
  \]
  Then for any leaf $i \in V$ one simply has $K_i = K$. We prove that the constructed bodies are still large: \\
  {\bf Claim I.} \emph{For all $i \in V$ one has $\gamma_n(K_i) \geq 1-\frac{|D_i|}{2|E|}$.} \\
  {\bf Proof of Claim I.} We prove the claim by induction. The claim is true for any leaf $i \in V$
  because in that case $|D_i|=1$ and $\gamma_n(K_i) = \gamma_n(K) \geq 1-\frac{1}{2|E|}$. Now consider
  any interior node $i$ and assume the claim is true for all its children. In particular for each $j \in C_i$
  we have $\gamma_n(K_j) \geq 1-\frac{|D_j|}{2|E|} \geq \frac{1}{2}$ by induction and the fact for any non-root
  node $|D_j| \leq |E|$. Hence 
  \[
   \gamma_n(K_j * \beta v_{\{i,j\}}) \stackrel{\textrm{Thm~\ref{thm:BanaszczykStretchedBodyLB}}}{\geq} \gamma_n(K_j) \stackrel{\textrm{induction}}{\geq} 1-\frac{|D_j|}{2|E|}. 
  \]
  Then
  \begin{eqnarray*}
    \gamma_n(K_i) &=& \gamma_n\Big( \Big(\bigcap_{j \in C_i} (K_j * \beta v_{\{i,j\}})\Big) \cap K \Big) \\
    &\stackrel{\textrm{union bound}}{\geq}& 1 - \sum_{j \in C_i} \underbrace{\gamma_n(\setR^n \setminus (K_j * \beta v_{\{i,j\}}))}_{\leq \frac{|D_j|}{2|E|}} - \underbrace{\gamma_n(\setR^n \setminus K)}_{\leq \frac{1}{2|E|}} \\
    &\geq& 1 - \frac{1}{2|E|} \underbrace{\Big(1+\sum_{j \in C_i} |D_j|\Big)}_{=|D_i|} = 1-\frac{|D_i|}{2|E|}. \quad \qed 
  \end{eqnarray*}
{\bf Claim II.} \emph{There are signs $x \in \{-1,1\}^E$ so that $\sum_{e \in P_i} x_{e} v_{e} \in \frac{1}{\beta} K_i$ for all $i \in V \setminus \{r\}$.} \\
{\bf Proof of Claim II.} We construct the signs in increasing distance to the root. If $i$ is a child of the root itself, then
the claim is true as $\gamma_n(K_i) \geq \frac{1}{2}$ and $\|v_{\{r,i\}}\|_2 \leq 1$.
Suppose for some node $i \in V \setminus \{r\}$ we have determined all the signs $x_{e} \in \{ -1,1\}$ with $e \in P_i$
and in particular $a := \beta\sum_{e \in P_i} x_{e} v_{e} \in K_i$. Consider any child $j \in C_i$. 
Then
\[
  a \in K_i \stackrel{\textrm{Def }K_i}{\subseteq} K_{j} * \beta v_{\{i,j\}} \stackrel{\textrm{Thm~\ref{thm:BanaszczykStretchedBodyLB}}}{\subseteq} (K_{j} + \beta v_{\{i,j\}}) \cup (K_{j}-\beta v_{\{i,j\}})
\]
Then we may pick a sign $x_{\{i,j\}} \in \{ -1,1\}$ so that $a + \beta x_{\{i,j\}} v_{\{i,j\}} \in K_{j}$. \qed

Then the overall claim follows from the fact that $K_i \subseteq K$ for all $i \in V \setminus \{r\}$ and for the root itself one trivially has $\bm{0} \in K$.
\end{proof}

\section{The body of subgaussian distributions}

In this section, we introduce the target body $K$ that we will use and we will prove that
its Gaussian measure is close enough to $1$ to make the argument work.
For $n,N \in \setN$ and a small $\delta > 0$, we define 
\begin{equation} \label{eq:DefK}
  K := \Big\{ (y^{(1)},\ldots,y^{(N)}) \in \setR^{Nn} \mid \| Y\|_{\psi_2,\infty} \le 2+\delta\textrm{ where }Y \sim \{y^{(1)}, \dots, y^{(N)}\}\Big\}.
\end{equation}
Intuitively, the vectors in $K$ consist of $N$ many blocks of dimension $n$ with the property that a uniform
random block generates a subgaussian random vector.
Since $\| \cdot \|_{\psi_2,\infty}$ is a norm, $K$ is a symmetric convex body. The main result for this section will
be that $K$ has a large Gaussian measure if $N$ is large enough.

\begin{proposition} \label{prop:GaussianMeasureOfK}
  For any $\delta > 0$,
  there is a constant $C_\delta >0$ so that for all $n,N \in \setN$ one has $\gamma_{Nn}(K) \geq 1- \frac{C_\delta^n}{N^{1+\delta}}$. 
\end{proposition}
We first show how to control the deviation in a single direction $w$. Note that random variables for the form $X := \exp(g^2)$ with $g \sim N(0,\sigma^2)$ and $\sigma>0$ have \emph{heavy tails}, which means
they do not possess finite exponential moments (i.e. $\E[e^{\lambda X}]=\infty$ for all $\lambda>0$). That implies in particular that standard Chernoff bounds cannot be used to derive concentration.
\begin{lemma} \label{lem:SinglewBound}
  For any $C > 4$ and $\delta' \in \Big(0, \frac{C}{4} - 1\Big)$, there is a $C'>0$ so that for any unit vector $w \in S^{n-1}$, the set
  \[
  K_w := \Big\{ (y^{(1)},\ldots,y^{(N)}) \in \setR^{Nn} \mid \E_{\ell \sim [N]}\big[ \exp\big(\tfrac{1}{C} \left<w,y^{(\ell)}\right>^2\big) \big] \leq 2\Big\}
\]
satisfies $\gamma_{Nn}(K_w) \ge 1 - \frac{C'}{N^{C/4 - \delta'}}$.
\end{lemma}

\begin{proof}
  Since  $y^{(\ell)} \sim N(\bm{0},I_n)$, the inner product satisfies $\langle w, y^{(\ell)} \rangle \sim N(0,1)$. 
  We define $Y_\ell := \exp\big(\tfrac{1}{C} \left<w,y^{(\ell)}\right>^2\big)$ with $\mu_C := \E[Y_\ell] = \Big(1-\frac{2}{C}\Big)^{-1/2}$ by Lemma~\ref{lem:RosenthalInequality}
  and the centered random variable $X_\ell := Y_\ell - \mu_C$, so that for any $p > 2$ we have by Markov
\begin{align*} \Pr \left[Y_1 + \dots + Y_N > 2N\right] & = \Pr \left[X_1 + \dots + X_N > (2 - \mu_C) \cdot N \right] \\ &  \le \Pr \left[|X_1 + \dots + X_N|^p > (2 -\mu_C)^p \cdot N^p \right]\\ & \le \frac{\E[|X_1 + \dots + X_N|^p]}{(2 -\mu_C)^p N^p}. \end{align*} 
Since $X_\ell$ are mean zero independent random variables with $p$-th moment
\[
\E[|X_\ell|^p] = E[|Y_\ell - \mu_C|^p] \le \E[Y_\ell^p] + \mu_C^p = \mu_{C/p} + \mu_C^p,
\]
which is a finite constant $M_p$ for $p < C/2$, we may apply Lemma~\ref{lem:RosenthalInequality} to obtain

\[ \E[|X_1 + \dots + X_N|^p] \le \Big(2^{p} \cdot \max\{M_p \cdot N^{1/p}, M_2 \cdot N^{1/2}\}\Big)^p \le C_p N^{p/2}. \]

For $p := C/2 - 2\delta'$, we conclude $\Pr \left[Y_1 + \dots + Y_N > 2N\right] \le \frac{C_p}{(2-\mu_C)^p N^{p/2}} = \frac{C'}{N^{C/4-\delta'}}$ for some constant $C'$ depending only on $C$ and $\delta'$. \end{proof}

\begin{remark}
Lemma~\ref{lem:SinglewBound} is tight in the sense that one cannot hope for a lower bound $\gamma_{Nn}(K_w) \ge 1 - O(1/N)$ for any $C < 4$. Indeed, for $Y_\ell := \exp\big(\tfrac{1}{C} \left<w,y^{(\ell)}\right>^2\big)$, one has
\[ \Pr \left[Y_1 + \dots + Y_N > 2N\right] \ge \Pr \Big[\max_{\ell \in [N]} Y_\ell > 2N\Big] = 1 - \Big(\Pr_{g \sim N(0,1)} \Big[g \le \sqrt{C \log(2N)}\Big]\Big)^N, \]
which is $\Omega\Big(\tfrac{1}{N^{C/2 - 1} \sqrt{\log(2N)}}\Big)$ for any $C \ge 2$ by using standard Gaussian tail bounds.
\end{remark}

\begin{remark}
In fact, the lower bound above is tight up to constants for $C > 8/3$ (which ensures $2 > \mu_C$). This is a consequence of Theorem 3.1.6 in~\cite{borovkov_borovkov_2008} (see Equation 2.1.6 for context). The proof follows from a truncation of the random variables and a careful estimation of the resulting moment generating function. Thus there is a stronger version of Lemma~\ref{lem:SinglewBound} which also works for $C = 4$, although it does not seem to lead to an improvement of Proposition~\ref{prop:GaussianMeasureOfK}.
\end{remark}

Obviously one cannot combine Lemma~\ref{lem:SinglewBound} with a union bound over the infinitely
many vectors $w \in S^{n-1}$, but it is a standard argument that in such cases it suffices to control the deviation in
directions of an $\varepsilon$-net, see again the textbook of Artstein-Avidan, Giannopoulos and Milman~\cite{AsymptoticGeometricAnalysis-Book2015}.
\begin{lemma} \label{lem:EpsNetSubgaussian}
  Let $W \subseteq S^{n-1}$ be an $\varepsilon$-net for $0<\varepsilon<1$ 
  and let $X \in \setR^n$ be a random vector so that $\|\left<X,w\right> \|_{\psi_2} \leq 1$ for all $w \in W$. Then $\|X\|_{\psi_2,\infty} \leq \frac{1}{1-\varepsilon}$.
\end{lemma}
\begin{proof}
  Fix $w \in S^{n-1}$. By Lemma~\ref{lem:NetImpliesSmallConicComb} there are coefficients $\lambda \in \setR_{\geq 0}^W$
  with $w = \sum_{u \in W} \lambda_u u$ and $\|\lambda\|_1 \leq \frac{1}{1-\varepsilon}$. Then as $\| \cdot \|_{\psi_2}$
  is a norm we have
  \[
   \|\left<X,w\right>\|_{\psi_2} = \Big\| \left<X,\sum_{u \in W} \lambda_u u\right>\Big\|_{\psi_2} \leq \sum_{u \in W} \lambda_u \underbrace{\|\left<X,u\right>\|_{\psi_2}}_{\leq 1} \leq \frac{1}{1-\varepsilon}.\qedhere
  \]
\end{proof}

Now we are ready to prove Proposition~\ref{prop:GaussianMeasureOfK}:

\begin{proof}[Proof of Prop.~\ref{prop:GaussianMeasureOfK}]
  Set $C := (2+\delta)^2 - \delta^2/2$, $\eps := 1 - \frac{\sqrt{C}}{2+\delta}$ and let $W$ be an $\eps$-net of $S^{n-1}$ of size at most $(\frac{3}{\eps})^n$ (see Lemma~\ref{lem:SizeOfEpsilonNet}).
  Applying Lemma~\ref{lem:SinglewBound} with $C$ and $\delta' := \delta^2/8$ (chosen so that $C/4 - \delta' = 1+\delta$) and taking the union bound over all $w \in W$,
  it follows that for $y^{(\ell)} \sim N(\bm{0}, I_n)$ and $Y \sim \{y^{(1)},\dots, y^{(N)}\}$ one has
  \[
    \Pr[\| \langle Y, w\rangle \|_{\psi_2} \le \sqrt{C} \; \forall w \in W] \geq 1 - \Big(\frac{3}{\eps}\Big)^n \cdot \frac{C'}{N^{1+\delta}}
  \]
  for some $C'>0$.
  If this event happens, then by Lemma~\ref{lem:EpsNetSubgaussian} also $\|Y\|_{\psi_2,\infty} \leq 2+\delta$.
\end{proof}


\section{Subgaussianity over trees}

Recall that in Theorem~\ref{thm:TreeVectorBalancing} we have proven that we can find a \emph{single coloring} so that
all vectors on any path in a tree are balanced into a (large enough) convex set $K$. We can use the same argument
to find a \emph{distribution} that keeps all those vector sums subgaussian.
\begin{theorem} \label{thm:TreeDistribution} There exists a constant $\gamma < 10$ such that the following holds. Let $\pazocal{T} = (V,E)$ be a tree with a distinguished root, where each edge $e \in E$ is assigned a vector $v_e \in \setR^n$ with $\|v_e\|_2 \leq 1$.
Then there is a distribution $\pazocal{D}$ over $\{ -1,1\}^E$ so that for $x \sim \pazocal{D}$, 
$
  \sum_{e \in P_i} x_e v_e 
$
is $\gamma$-subgaussian for every $i \in V$ where $P_i \subseteq V$ are the edges on the path from the root to $i$. 
\end{theorem}
\begin{proof}
  Let $N \in \setN$ be a large enough parameter that we determine later.
We replace each edge $e \in E$ by a path consisting of $N$ many edges where the vectors assigned to edges are
\[
  v_e^{(\ell)} := (\bm{0},\ldots,\bm{0},v_e,\bm{0},\ldots,\bm{0}) \in \setR^{Nn} \quad \forall \ell=1,\ldots,N
\]
Note that still $\|v_e^{(\ell)}\|_2 \leq 1$. We call the new tree $\pazocal{T}' = (V',E')$ and note that $|E'| = N\cdot |E|$.
We define $K$ as in \eqref{eq:DefK} and by Proposition~\ref{prop:GaussianMeasureOfK} we have 
\[
  \gamma_{Nn}(K) \geq 1-\frac{C_\delta^n}{N^{1+\delta}} \geq 1-\frac{1}{2|E'|},
\]
choosing $N := (2|E| \cdot C_\delta^n)^{1/\delta}$. 
We apply Theorem~\ref{thm:TreeVectorBalancing} to the tree $\pazocal{T}'$ in order to obtain signs $(x_e^{\ell})_{e \in E, \ell \in [N]} \in \{ -1,1\}^{|E| \cdot N}$ so that
\[
 \sum_{\ell=1}^N \sum_{e \in P_i} x_e^{(\ell)} v_{e}^{(\ell)} \in \alpha K
\]
for all $i \in V$ where $\alpha < 5$. If we draw an index $\ell \sim [N]$ uniformly and set $X := (x^{(\ell)}_1,\ldots,x^{(\ell)}_{|E|})$, then by
construction of $K$ one has
\[
    \Big\|\sum_{e \in P_i} X_ev_e\Big\|_{\psi_2,\infty} \leq \alpha \cdot (2 + \delta) < 10 \;\; \forall i \in V
  \]
  if we choose $\delta>0$ small enough.
\end{proof}
Note that the produced distribution $\pazocal{D}$ is indeed the uniform distribution over a multiset of $(2|E| \cdot C_\delta^n)^{1/\delta}$ sign vectors. We remark that even if $\pazocal{T}$ is a path with the $n$ vectors $e_1,\ldots,e_n$,
any $O(1)$-subgaussian distribution over sign vectors needs to have support at least $2^{\Omega(n)}$. 

We should also point out that our proof strategy is related to an argument by Raghu Meka to use Banaszczyk's
Theorem~\ref{thm:BanaszczykStretchedBodyLB} to prove the existence of a good SDP solution. This was reported in the work of \cite{FlowTimeAndPrefixBeckFialaSTOC2022} in the context of an SDP solution
with small discrepancy for all prefixes. After completing a preliminary draft, we learned the existence of an unpublished manuscript~\cite{NikolovPersonal}
proving Theorem~\ref{thm:TreeDistribution} for the special case where $\pazocal{T}$ is a path.

\section{Existence of an online algorithm}

Finally, we use the tree subgaussianity from Theorem~\ref{thm:TreeDistribution} to show the existence of an online algorithm which maintains subgaussian prefixes.
{}
\begin{theorem} \label{thm:fixedTimeStepT}
For every $T \in \mathbb{N}$, there exists a randomized online algorithm which, upon receiving a vector $v_i \in \setR^n$ with $\|v_i\|_2 \le 1$ for each $i \in [T]$, outputs a random sign $x_i \in \{-1,1\}$ so that the prefix sum $\sum_{j=1}^i x_j v_j$ is $10$-subgaussian. The algorithm runs in time $\exp(T^{CnT})$ for some universal constant $C > 0$.
\end{theorem}

\begin{proof}
  We write the constant from Theorem~\ref{thm:TreeDistribution} as $10-\delta$ for some $\delta>0$.
  Let $W$ be an $\eps$-net of $S^{n-1}$ of size $(3/\eps)^n$ where $\varepsilon := \frac{\delta}{10T}$. We consider a rooted tree $\pazocal{T} = (V,E)$ of depth $T$ so that all non-leaf nodes have $|W|$ children, labeled with each of the vectors in $W$. 
  By Theorem~\ref{thm:TreeDistribution}, there exists a distribution $\pazocal{D}$ over signs $x \in \{ -1,1\}^E$ that is $(10-\delta)$-subgaussian.
 
We now describe the randomized online algorithm. First, we sample random signs $x \sim \pazocal{D}$ and keep track of a position $p \in V$ in the tree, initially set to be the root. Now for each incoming vector $v_i \in \setR^n$, we output the sign corresponding to an edge $(p,p')$ labelled with $v_i'$ that satisfies $\|v_i - v'_i\|_2 \le \eps$ and set $p := p'$. It remains to argue that the resulting prefix sums are indeed $10$-subgaussian. Indeed, for any $i$ one has
\begin{align*}
\Big\|\sum_{j=1}^i x_j v_j\Big\|_{\psi_2,\infty} & \le \underbrace{\Big\|\sum_{j=1}^i x_j v'_j\Big\|_{\psi_2,\infty}}_{\le 10-\delta} + \Big\|\sum_{j=1}^i x_j (v_j - v_j')\Big\|_{\psi_2,\infty} \\ & \le 10-\delta + \sum_{j=1}^i \sup_{w \in S^{m-1}} \underbrace{| \langle  v_j - v'_j,w\rangle|}_{\le \eps} \cdot \underbrace{\|x_j\|_{\psi_2}}_{\le 2}  \\ & \le 10 - \delta + T \cdot 2\varepsilon < 10.
\end{align*}

Next, we discuss the running time of this procedure. The total number of edges in $\pazocal{T}$ equals $|E| = |W| + \dots + |W|^T \le (2|W|)^T \le O(T/\delta)^{nT}$.
Reinspecting the proof of Theorem~\ref{thm:TreeDistribution} we recall that
the $(10-\delta)$-subgaussian distribution $\pazocal{D}$ is constructed from a proper sign vector
for the tree $\pazocal{T}' = (V',E')$ that has  $|E'| = N\cdot |E|$ many edges where  $N := (2|E| \cdot C_\delta^n)^{1/\delta}$. The number of candidate sign vectors to be tried out
is then $2^{|E'|} \leq \exp((T/\delta)^{O(nT)})$, and we can verify the subgaussianity of each by computing the subgaussian norm of each root-vertex path sum over inner products with all vectors from $W$ in time $\exp(T^{CnT})$; then Proposition~\ref{lem:EpsNetSubgaussian} guarantees that if all such inner products are $(10-\delta)$-subgaussian, then the root-vertex path sums are also $\frac{10-\delta}{1-\eps} \le \frac{10-\delta}{1-\delta/10} = 10$-subgaussian. 
\end{proof}

\section{An online algorithm independent of the input length}

In the previous section we argued that for every number $T$ of vectors, there is an online algorithm that balances $T$
vectors. In this section, we prove that in fact, there has to be a \emph{single} online algorithm that balances
any sequence of vectors without knowing the number vectors beforehand.
Note that just using Theorem~\ref{thm:fixedTimeStepT} as a black box, the online algorithm
to balance $T$ vectors could be very different from the algorithm to balance $T'<T$
vectors. But of course one could have used the algorithm that worked for $T$ vectors also
to balance just $T'$ vectors. By using a compactness argument we will argue that there is indeed a single algorithm.

We fix some $\varepsilon > 0$ and set $\varepsilon_i := \varepsilon 2^{-i}$ for all $i \geq 1$.
Let $W_i \subseteq S^{n-1}$ be an $\varepsilon_i$-net. Let $\pazocal{T}_i = (V_i,E_i)$ be a tree of height $i$ with a distinguished root $r$ where for all $j \in \{ 1,\ldots,i\}$, each node at depth $j-1$ has $|W_j|$
many outgoing edges, one labelled with each vector from $W_j$. In other words, any root-leaf path in $\pazocal{T}_i$
corresponds to a sequence $(v_1,\ldots,v_i)$ with $v_j \in W_j$ for $j =1,\ldots,i$.
We write $\pazocal{T}^* = (V^*,E^*)$ as the infinite tree constructed in
the same manner, i.e. for all $j \geq 1$, nodes at distance $j-1$ to the root have $|W_j|$ children.
We say that a distribution $\pazocal{D}_i$ over signs $\Omega_i := \{ -1,1\}^{E_i}$ is \emph{$c$-subgaussian for $\pazocal{T}_i$} if $\|\sum_{P} x_ev_e\|_{\psi_2} \leq c$
for $x \sim \pazocal{D}_i$ and every path $P$ in $\pazocal{T}_i$ starting at the root.  First we prove that
the distributions for different height trees can be chosen in a consistent way:

\begin{lemma} \label{lem:ProjectedDistributions}
  There is a constant $\gamma < 10$ so that for every $\varepsilon > 0$
   there is a family of distributions $\{\pazocal{D}_i^*\}_{i \geq 1}$ where each $\pazocal{D}_i^*$ is a distribution over $\Omega_i$
  that is $\gamma$-subgaussian for $\pazocal{T}_i$ and $\pazocal{D}_i^* =\Pi_{\Omega_i}(\pazocal{D}_{i+1}^*)$ for all $i \geq 1$.
\end{lemma}
Here $\Pi_{\Omega_i}(\pazocal{D}_{i+1}^*)$ denotes the projection (or the marginals) of $\pazocal{D}_{i+1}^*$ on $\Omega_i$.
\begin{proof}
Consider the sequence of $\gamma$-subgaussian distributions $\pazocal{D}_i$ for $\pazocal{T}_i$ (which depend on $\varepsilon$) given by Theorem~\ref{thm:TreeDistribution} as points in the metric space $\Delta_i := \{x \in \setR_{\ge 0}^{\Omega_i} : \|x\|_1 = 1\}$. Since $\Delta_i$ is closed and bounded, it is compact. For every $i' \le i$, since $\pazocal{T}_{i'} \subseteq \pazocal{T}_i$, it follows that $\Pi_{\Omega_{i'}} (\pazocal{D}_i)$ is also $\gamma$-subgaussian.

We construct $\pazocal{D}^*_i$ inductively. First consider the infinite sequence of distributions $\Pi_{\Omega_1} (\pazocal{D}_i)$, all of which are $\gamma$-subgaussian for $\pazocal{T}_1$. Since $\Delta_i$ is compact, there exists a subsequence of indices $\{k_{j, 1}\}_{j \ge 1}$ so that $\pazocal{D}^*_1 := \lim_{j \to \infty} \Pi_{\Omega_1} (\pazocal{D}_{k_{j,1}})$ exists. By continuity, $\pazocal{D}^*_1$ is also $\gamma$-subgaussian over $\pazocal{T}_1$. 

Now assume that we have constructed distributions $\pazocal{D}^*_\ell$ for $1 \le \ell \le i$ that are $\gamma$-subgaussian for $\pazocal{T}_\ell$ with $\pazocal{D}_\ell^* = \Pi_{\Omega_\ell}(\pazocal{D}_{\ell+1}^*)$ for $1\le \ell < i$, as well as an infinite sequence of indices $\{k_{j,i}\}_{j \ge 1}$ which again satisfy $\pazocal{D}_\ell^* = \lim_{j \to \infty} \Pi_{\Omega_\ell}(\pazocal{D}_{k_{j,i}})$ for $\ell \le i$. We may drop a prefix if needed so that $k_{j,i} \ge i+1$ for all $j \ge 1$. By compactness of $\Delta_{i+1}$, there exists a subsequence of indices $\{k_{j,i+1}\}_{j \ge 1}$ so that $\pazocal{D}^*_{i+1} := \lim_{j \to \infty} \Pi_{\Omega_{i+1}} (\pazocal{D}_{k_{j,i+1}})$ exists; the previous limits remain the same. Again by continuity of the subgaussian norm, $\pazocal{D}^*_{i+1}$ is $\gamma$-subgaussian over $\pazocal{T}_{i+1}$, and also $\Pi_{\Omega_i} (\pazocal{D}^*_{i+1}) =\Pi_{\Omega_i} (\lim_{j \to \infty} \Pi_{\Omega_{i+1}} (\pazocal{D}_{k_{j,i+1}})) =  \lim_{j \to \infty} \Pi_{\Omega_i} (\pazocal{D}_{k_{j,i+1}}) = \pazocal{D}^*_i $.\end{proof}

\begin{theorem} There exists an online algorithm which, upon receiving a vector $v_i \in \setR^n$ with $\|v_i\|_2 \le 1$, outputs a random sign $x_i \in \{-1,1\}$ so that the prefix sum $\sum_{j=1}^i x_j v_j$ is $10$-subgaussian.
\end{theorem}

\begin{proof}
  Fix a small enough $\delta > 0$ so that Lemma~\ref{lem:ProjectedDistributions} works with
  constant $10-\delta$. Let $\pazocal{D}^*_i$ be $(10-\delta)$-subgaussian distributions provided by Lemma~\ref{lem:ProjectedDistributions} where we choose $\eps := \frac{\delta}{4}$. We keep track of a position $p \in V$ in the tree, initially set to be the root. For each  incoming vector $v_i \in \setR^n$, we sample a random sign $x_i$ from $\pazocal{D}^*_i$ \emph{conditioned} on the previous signs $x_1, \dots, x_{i-1}$ on edges already visited, corresponding to an edge $(p,p')$ labeled with a vector $v'_i$ that satisfies $\|v_i - v'_i\|_2 \le \eps_i$, and set $p := p'$. It remains to argue that the resulting prefix sums are $10$-subgaussian. Indeed, since $\pazocal{D}^*_j = \Pi_{\Omega_j}(\pazocal{D}^*_i)$ for all $j < i$, the distribution of the signs $(x_1, \dots, x_i)$ is equivalent to drawing all of them from $\pazocal{D}^*_i$. Then, as in Theorem~\ref{thm:fixedTimeStepT},
\begin{align*}
\Big\|\sum_{j=1}^i x_j v_j\Big\|_{\psi_2,\infty} & \le \underbrace{\Big\|\sum_{j=1}^i x_j v'_j\Big\|_{\psi_2,\infty}}_{\le 10-\delta} + \Big\|\sum_{j=1}^i x_j (v_j - v_j')\Big\|_{\psi_2,\infty} \\ & \le 10-\delta + \sum_{j=1}^i \sup_{w \in S^{m-1}} \underbrace{| \langle v_j - v'_j,w\rangle|}_{\le \eps_j} \cdot \underbrace{\|x_j\|_{\psi_2}}_{\le 2}  \\ & \le 10- \delta + 2\eps < 10. \qedhere
\end{align*}
\end{proof}

\section{Applications}

In this section, we show a few direct consequences of our main Theorem~\ref{thm:MainResultI}.

\begin{proof}[Proof of Theorem~\ref{thm:MainResultBanaszczyk}]

Both items will follow from Lemma~\ref{lem:TalagrandComparison}. For the first item, note that by Lemmas 26 and 27 in~\cite{DBLP:conf/approx/DadushGLN16}, it follows that any symmetric convex body $K$ with $\gamma_n(K) \ge 1/2$ satisfies $E_{g \sim N(\bm{0},I_n)} [\|g\|_K] \lesssim 1$ and $\mathrm{inradius}(K) \gtrsim 1$, so it suffices to apply Lemma~\ref{lem:TalagrandComparison} with $\delta:= \frac{1}{2}$. 

For the second item, any symmetric convex body $K$ with $\gamma_n(K) \ge 1 - \frac{1}{2T}$ once again satisfies $\E_{g \sim N(\bm{0},I_n)} [\|g\|_K] \lesssim 1$, and moreover $K$ must contain a $\Omega(\sqrt{\log T})$-radius ball, for otherwise it would be contained in a strip of Gaussian measure less than $1 - \frac{1}{2T}$. Therefore $\max_{x \in K^\circ} \|x\|_2 \lesssim 1/\sqrt{\log T}$ and we may apply Lemma~\ref{lem:TalagrandComparison} with $\delta := \frac{1}{2T}$ together with the union bound over the $T$ prefixes.
\end{proof}

\begin{proof}[Proof of Corollary~\ref{thm:MainResultLp}]

Let $d \le \min(n,T)$ denote the dimension of the linear span $U$ of $v_1, \dots, v_T$. By Corollary 19 in~\cite{BGMN2005}, \[\E_{g \sim N(\bm{0}, I_n)}[\|g\|_{U \cap B^n_p}] \le \E_{g \sim N(\bm{0}, I_d)}[\|g\|_{B^d_p}] \stackrel{\textrm{Jensen}}{\le} \E_{g \sim N(\bm{0}, I_d)}[\|g\|_{B^d_p}^p]^{1/p} \le \sqrt{p} \cdot d^{1/p}. \] Here we also used that $\E_{g \sim N(0,1)}[g^p] \le p^{p/2}$. By Lemma~\ref{lem:TalagrandComparison}, \begin{align*} \|X\|_p = \|X\|_{U \cap B^n_p} & \lesssim \E_{g \sim N(\bm{0}, I_n)}[\|g\|_{U \cap B^n_p}] + \sqrt{\log(1/\delta)} \cdot \underbrace{1/\mathrm{inradius}(B^n_p)}_{=1} \\ & \le \sqrt{p} \min(n,T)^{1/p} + \sqrt{\log(1/\delta)} \end{align*}
with probability at least $1-\delta$, as claimed. The second item follows from a union bound over $T$ such events and that $\sqrt{\log(T/\delta)} \le \sqrt{\log T} + \sqrt{\log(1/\delta)}$. 

For the $\ell_\infty$ bounds, we instead apply Theorem 9 in~\cite{BGMN2005} which together with Lemma 26 in~\cite{DBLP:conf/approx/DadushGLN16} gives $\E_{g \sim N(\bm{0}, I_n)}[\|g\|_{U \cap B^n_\infty}] \lesssim \sqrt{\log \min(n,T)}$, so that Lemma~\ref{lem:TalagrandComparison} analogously yields both claims.
\end{proof}

\begin{proof}[Proof of Corollary~\ref{cor:edgeOrientation}]
It suffices to construct a vector $v_e := e_u - e_v \in \setR^{|V|}$ for every edge $e = \{u,v\}$; then signs correspond to an orientation. Since such vectors have constant $\ell_2$ norm, the claim follows directly from Corollary~\ref{thm:MainResultLp}(d).
\end{proof}

\section{Online discrepancy lower bound}

We show that the bound in Corollary~\ref{thm:MainResultLp}(d) is tight up to a constant:

\begin{proof} [Proof of Theorem~\ref{thm:OnlineDiscLB}]

We may also assume that $n = 2$ as otherwise we may pad the construction with zeros. By Yao's minimax principle~\cite{10.1109/SFCS.1977.24}, it suffices to show a strategy for the adversary against deterministic online algorithms. 

Split the time horizon into $T/k$ blocks of length $k := \lceil \tfrac{1}{2} \log T\rceil$. In each block, the adversary samples signs $y \sim \{-1,1\}^k$ uniformly at random, and outputs unit vectors $v_1, \dots, v_k \in \setR^2$ so that for each $i \in [k]$, $v_i$ is orthogonal to $\sum_{j < i} y_j v_j$. The online algorithm then outputs deterministic signs $x \in \{-1,1\}^k$ where $x_i$ only depends on $y_j, v_j$ for $j < i$. Since the sign $y_i$ is independent of the sign $x_i$ given by the online algorithm, the sign vectors $x$ and $y$ will be equal with probability $2^{-k}$. In this case, $\|\sum_{i=1}^k x_i v_i\|_2 = \sqrt{k}$. The probability that this happens in \emph{at least one} block is $1 - (1-2^{-k})^{T/k} = 1 - 2^{-T^{\Omega(1)}}$. Then one of the two prefixes induced by that block will have $\ell_2$ norm at least $\frac{1}{2} \sqrt{k}$ and $\ell_\infty$ norm at least $\frac{1}{\sqrt{2}} \cdot \frac{1}{2} \sqrt{k} \gtrsim \sqrt{\log T}$.
\end{proof}

\section{Open problems}

We mention two other settings for which the optimal discrepancy bound against oblivious adversaries remains open:

\begin{conjecture} Does there exist an online algorithm that for any sequence of vectors $v_1, \dots, v_n \in \setR^n$ with $\|v_i\|_\infty \le 1$, arriving one at a time, decides random signs $x_1, \dots, x_n \in \{-1,1\}$ so that $\|\sum_{i=1}^n x_i v_i\|_\infty \le O(\sqrt{n})$ with high probability?
\end{conjecture}

Corollary~\ref{thm:MainResultLp}(c) gives a bound of $O(\sqrt{n \log n})$. Bansal and Spencer have settled the case where the vectors are chosen uniformly at random from $\{-1,1\}^n$~\cite{BansalSpencer2020}. Note that there is a $\Omega(\sqrt{n})$ lower bound even in the offline setting and a $\Omega(\sqrt{\log T})$ lower bound from Theorem~\ref{thm:OnlineDiscLB} (thus we restrict to $T = n$).

\begin{conjecture}
Does there exist an online algorithm that for any sequence of vectors $v_1, \dots, v_T \in \setR^n$, each with two nonzero coordinates (one equal to 1 and the other -1) and arriving one at a time, decides random signs $x_1, \dots, x_T \in \{-1,1\}$ so that $\|\sum_{i=1}^t x_i v_i\|_\infty \le O(\sqrt[3]{\log T})$ for all $t \in [T]$ with high probability?
\end{conjecture}

Corollary~\ref{cor:edgeOrientation} gives an upper bound of $O(\sqrt{\log T})$ and there is a $\Omega(\sqrt[3]{\log T})$ lower bound~\cite{AJTAI1998306}~\cite{fiat2017carpooling}.

Finally, we ask for a polynomial time algorithm for Theorem~\ref{thm:MainResultI}.

\begin{conjecture}
Does there exist a polynomial time online algorithm that against any oblivious adversary, for any sequence of vectors $v_1,\ldots,v_T \in \setR^n$ with $\|v_i\|_2 \leq 1$, decides random signs $x_1,\ldots,x_T \in \{ -1,1\}$ so that for every $t \in [T]$, the prefix sum $\sum_{i=1}^t x_iv_i$ is $O(1)$-subgaussian?
\end{conjecture}
\bibliographystyle{alpha}
\bibliography{onlineVectorBalancing}

\end{document}